\documentclass[aps,prd,reprint,superscriptaddress,nofootinbib,twocolumn]{revtex4-1}

\usepackage[english]{babel}

\usepackage{amsmath}
\usepackage{amsthm}
\usepackage{latexsym}
\usepackage{graphicx}

\usepackage{physics}
\usepackage{amssymb}
\usepackage{comment}
\usepackage{flexisym}
\usepackage{breqn} %%% For line breaking with \left \right brackets

\PassOptionsToPackage{breaklinks}{hyperref}
\usepackage{xurl}
\usepackage[colorlinks=true, allcolors=blue]{hyperref}

\usepackage{mathtools}  
\usepackage[normalem]{ulem}
\usepackage{tikz}
\usetikzlibrary{shapes,shadows,arrows}
\usepackage{tikz-cd}
\usepackage{amssymb}
\usetikzlibrary{calc}
\usetikzlibrary{decorations.pathmorphing}

\tikzset{curve/.style={settings={#1},to path={(\tikztostart)
    .. controls ($(\tikztostart)!\pv{pos}!(\tikztotarget)!\pv{height}!270:(\tikztotarget)$)
    and ($(\tikztostart)!1-\pv{pos}!(\tikztotarget)!\pv{height}!270:(\tikztotarget)$)
    .. (\tikztotarget)\tikztonodes}},
    settings/.code={\tikzset{quiver/.cd,#1}
        \def\pv##1{\pgfkeysvalueof{/tikz/quiver/##1}}},
    quiver/.cd,pos/.initial=0.35,height/.initial=0}

\tikzset{tail reversed/.code={\pgfsetarrowsstart{tikzcd to}}}
\tikzset{2tail/.code={\pgfsetarrowsstart{Implies[reversed]}}}
\tikzset{2tail reversed/.code={\pgfsetarrowsstart{Implies}}}
\tikzset{no body/.style={/tikz/dash pattern=on 0 off 1mm}}
\tikzstyle{decision} = [diamond, draw, fill=white]
\tikzstyle{line} = [draw, -stealth, thick]
\tikzstyle{elli}=[draw, ellipse, fill=gray!20, minimum height=8mm, text width=5em, text centered]
\tikzstyle{block} = [draw, rectangle, fill=white, text width=8em, text centered, minimum height=15mm, node distance=10em]

\newcommand{\dif}{\dd} 
\newcommand{\R}{\mathbb{R}}

\newcommand{\rem}[1]{}
\newcommand{\der}[2][]{\frac{\dd {#1}}{\dd {#2}}}
\newcommand{\pder}[2][]{\frac{\partial {#1}}{\partial {#2}}}

\newtheorem{proposition}{Proposition}
\newtheorem{assumptions}{Assumptions}

\usepackage{csquotes}
\newcommand{\VSI}{Van Swinderen Institute for Particle Physics and Gravity,\\ University of Groningen,
Nijenborgh 4, 9747 AG Groningen, The Netherlands}

\newcommand{\BI}{Bernoulli Institute for Mathematics, Computer Science and Artificial Intelligence,\\ University of Groningen, Nijenborgh 9, 9747 AG Groningen, The Netherlands}

\begin{document}

\title{Bertrand's Theorem and the Double Copy of Relativistic Field Theories}
\author{Dijs de Neeling}
\affiliation{\VSI}
\affiliation{\BI}
\author{Diederik Roest}\affiliation{\VSI}
\author{Marcello Seri}
\affiliation{\BI}
\author{Holger Waalkens}
\affiliation{\BI}

\begin{abstract}
\noindent
Which relativistic field theories give rise to Kepler dynamics in the two-body problem? We consider a class of Hamiltonians that is the unique relativistic extension of the Kepler problem preserving its so(4) algebra, and have orbits related through time reparametrisation to orbits of the original Kepler problem. For three explicit examples, we give a natural interpretation in terms of spin-0,-1 and -2 interacting field theories in 5D. These are organically connected via the classical double copy, which therefore preserves maximal superintegrability.
\end{abstract}

\maketitle

\section{Introduction} 
\noindent
The Laser Interferometer Space Antenna (LISA) is expected to open the door to observations of binary black holes systems with long inspirals and extreme mass ratio inspirals, both of which can be approached fruitfully from an analytical point of view~\cite{PhysRevD.95.103012,Berry:2019wgg,Cheung:2023lnj}.
The non-relativistic limit of these binaries in the center of mass frame is the Kepler problem. The properties of the latter are, even among classical systems, very special. It is one of only two central potential problems whose bounded orbits are all closed, the other being the isotropic harmonic oscillator, a result known as Bertrand’s theorem and going back to 1873~\cite{Bertrand,Goldstein2001-ql}. This is a consequence of the large, `hidden' so(4) symmetry of the system, leading to maximally superintegrable dynamics.

General Relativity breaks this enhanced so(4) symmetry down to so(3) (or even further in case of spinning objects), as do most other relativistic theories. However, there exist relativistic dynamical systems with the same symmetry group as Kepler to all orders. The simplest of these are test-particle limits\footnote{This is to be expected: while it is far from trivial to write down the Hamiltonian for a multi-worldline system explicitly, the test-particle case can be viewed as a one-body problem, as the large mass is taken to be non-moving.} and include the 'branonium' systems identified in string theory \cite{Burgess:2003qv,Burgess:2003mm}. More recently, the same hidden symmetry was found in $\mathcal{N}=4$ super-Yang-Mills~\cite{Caron_Huot_2014,Alvarez-Jimenez:2018lff} and $\mathcal{N}=8$ supergravity~\cite{Caron_Huot_2019}. Additionally, in the latter theory, the two-body system is known to possess the symmetry to first order in the Post-Newtonian expansion, while it preserves the related non-precession of orbits to third order in the Post-Minkowskian expansion~\cite{Parra_Martinez_2020dzs}. Finally, also the Kaluza-Klein monopole scattering of~\cite{GIBBONS1986183,GIBBONS1987226} should be noted\footnote{Another relativistic so(4)-conserving two-body system is the time-asymmetric scalar-vector theory studied by~\cite{Duviryak1996}.}.   

As we have shown recently by explicit construction, Bertrand's theorem can be extended beyond classical central potentials: at least up to and including 5th PN order, there exists a unique class of relativistic two-body Hamiltonians with the appropriate classical limit that displays additional symmetry. Moreover, all elements of this class can be related to functions of the Kepler Hamiltonian through canonical transformations~\cite{deNeeling:2023egt}. An all-order argument was subsequently put forward by Davis and Melville in \cite{Davis_2023zqv} and will be outlined below.

In the present work, we further investigate this class of relativistic Hamiltonians possessing Keplerian symmetry to all orders, having the correct special-relativistic and non-relativistic limits. Our focus will be on three examples of the class, for which we propose a realisation in terms of fundamental forces. We find that all these can be lifted to 5D systems, and subsequently linked to classical theories via the Eisenhart lift. This provides for a geometric interpretation of the link\footnote{
We generalise this link between classical and relativistic dynamics to multi-center systems in the Appendix~\ref{sec: prop1}, along with a rigorous mathematical proof of the relation.} with Keplerian dynamics. 

The 5D formulation also allows for a connection between the three examples via the double copy. Known primarily as a relation of scattering amplitudes in gravity appearing as the `square' of gauge theory amplitudes~\cite{Bern:2010ue}, we will instead employ the so-called classical double copy, similarly linking background solutions of these theories~\cite{Monteiro:2014cda}. More specifically, we will use it to connect dynamical systems in the sense of~\cite{Gonzo:2021drq}. Interestingly, the three examples of relativistic dynamics with Keplerian dynamics are exactly linked in this way. This provides a precise sense in which hidden symmetries and maximal superintegrability\footnote{Another, less symmetric example is provided by the integrability of the Kerr space-time that also carries over~\cite{Volovich}.} carry over under the classical double copy.

\section{Relativistic Bertrand's Theorem}

\noindent
It is natural to wonder what lies beyond the celebrated Bertrand's theorem of classical mechanics, once moving into the realm of relativistic physics. For reasons that will be outlined in what follows, we will consider the following (implicit) definition for our class of relativistic Hamiltonians:
 \begin{align}
     H^2 = m^2 c^4 + p^2 c^2 - \left( \frac{1}{r} \right)  F\left(mc^2 + H\right) \,, \label{Ham}
 \end{align}
where $F\left(x\right)$ is a continuous and smooth function that encodes the nature of the attractive force. Note that this class of Hamiltonians reduces, in the absence of the attractive force ($F \rightarrow 0$), to the special relativistic free particle. When perturbatively including the attractive force, 
 \begin{align}
  F(x) = f_0 + f_1 x + \tfrac{1}{2} f_2 x^2 + \ldots \,,
 \end{align}
we see the coefficients must have dimensions $\left[f_i\right]=M^{(2-i)} L^{(3-i)}T^{(i-2)}$. Specializing to $f_2={8 GM}/{c^2}$ as the only nonzero component for instance, with $M$ some reference mass, the resulting Hamiltonian reads in a post-Newtonian expansion 
  \begin{align} 
  H = mc^2  &+ \frac{p^2}{2 m}- \frac{8GMm}{r}  \notag \\
  &  -\tfrac18 \frac{p^4}{m^3c^2} + \frac{32 G^2M^2m}{r^2c^2} +\mathcal{O}({\frac{1}{c^4}}) \,. 
 \end{align}
In the non-relativistic limit, the system therefore reduces to Kepler (up to a constant rest-mass energy).

The relativistic corrections of this class of Hamiltonians retain the special property that all bounded trajectories are closed orbits. At a heuristic level, this can be seen in the following way. As energy is conserved, the argument of the force term in~\eqref{Ham} will take a specific, constant value for a given trajectory. For this trajectory, the force is therefore given by some constant $F$, and the Hamiltonian reads
 \begin{align}
 H = mc^2 \sqrt{1 + 2 H_{\rm Kep} / mc^2} \,,   
 \end{align}
for a Kepler Hamiltonian
 \begin{align} \label{eq: Kep}
     H_{\rm Kep} = \frac{p^2}{2m} - \frac{1}{r} \frac{F}{2mc^2}\,. 
 \end{align}
It is therefore natural to expect that this class of systems has closed orbits. This is confirmed by the presence of a relativistic generalisation of the Laplace-Runge-Lenz vector, that takes the form
 \begin{align}
  \vec{A} = \vec{p} \times \Vec{L}  - \frac{1}{2c^2} F\left(mc^2 + H\right)  \hat{r} \,,
 \end{align}
in terms of angular momentum $\vec L = \vec r \times \vec p$ and the unit position vector $\hat r$.
This vector is responsible for the maximal superintegrability and hence closed orbits in the non-relativistic case. 

The dynamics of the above theories can be seen as a non-standard way of stacking Kepler energy levels: the effective gravitational coefficient is set by $F$ and hence is orbit-dependent. In a hypothetical solar system governed by such an attractive force, planets would fail to satisfy the universal Kepler's third law relation between period and radius. As discussed in more detail in the Appendix~\ref{sec: prop1}, a particular time reparametrisation of~\eqref{Ham} connects it to a new Kepler Hamiltonian, for which Kepler's third law does hold for all orbits with the same universal proportionality constant. 

One issue with the above heuristic argument is that it only applies to separate energy levels and not the complete phase space. It would therefore be desirable to have an alternative perspective, that applies to the union of all orbits. Indeed this appears to be possible: we have demonstrated (with an explicit construction) that this class of Hamiltonians is the unique extension of the Kepler system that combines the special relativistic and the non-relativistic limits (when $F \rightarrow 0$ and $c\to\infty$, respectively) that has the so(4) hidden symmetry and closed orbits, up to and including 5PN order. It is therefore the natural generalisation of Bertrand's theorem to the relativistic domain (at least for the Kepler system).
Moreover, this system is canonically equivalent to Kepler in a neighbourhood around an energy level, at least up to 5PN order~\cite{deNeeling:2023egt}. 

A closely related claim has been put forward by Davis and Melville in~\cite{Davis_2023zqv}, who employ a reorganisation of PN corrections into powers of $1/r$ multiplied by functions of the Kepler Hamiltonian. They argue that the function multiplying the $1/r$ part of this expansion can be made vanishing (their so-called `LRL gauge') by means of canonical transformations. Subsequently, they demonstrate that all higher-order powers of $1/r$ similarly have to vanish for the system to have closed orbits and hidden symmetries. This then implies that the class of Hamiltonians \eqref{Ham} is canonically conjugate to Kepler to all PN orders.

Having introduced the class of Hamiltonians as the relativistic generalisation of the Kepler system that has hidden symmetries, in the next sections we turn to their physical interpretation; in other words, which forces generate this kind of dynamics? We will show that there are three cases - where we take $F$ to be constant, linear or quadratic in its argument - where a natural interpretation in terms of 4D relativistic field theories presents itself. Moreover, all three cases allow for an uplift to a 5D system. Different reductions of this 5D perspective allow one to prove the canonical equivalence between the classical and relativistic systems to all orders in the PN expansion.

\section{General Relativity}

\noindent
We will first discuss the quadratic case, $F(x) = \frac{1}{2} f_2 x^2$, and its interpretation in terms of general relativity. As the constant $f_2$ must have the dimension of length, it is natural to suppose this is given in terms of a mass scale set by some reference mass M, such that $f_2=2r_M$, with $r_M=\frac{4GM}{c^2}$ (while the factor of $4$ will be clear in a moment). The Hamiltonian reads
 \begin{align} \label{eq:GRpHam}
     H^2 = m^2 c^4 + p^2 c^2 - \left( \frac{r_M}{r} \right) (m c^2 + H)^2 \,.
 \end{align}
It has been shown\footnote{Compare equation (4.22) in~\cite{deNeeling:2023egt} with $H\to H+m$.} that this system arises when an extremal particle (with equal mass and charge) is orbiting a specific background of the 4D Einstein-Maxwell-dilaton system, with bulk Lagrangian density
 \begin{align}
  \mathcal{L}= \frac{\sqrt{-g}}{16\pi}\left(\frac{c^3}{G}R-\frac{2c^3}{G}(\partial\phi)^2-\frac{e^{-2a\phi}}{c}F^2\right) \,,
 \end{align}
where we take the units such that the gauge potential $A_\mu$ has $\left[A^2\right]=M L T^{-2}$. For closed orbits, the dilaton coupling of this system has to take precisely the Kaluza-Klein value $a=\sqrt{3}$ that allows for the bulk Lagrangian to be uplifted to 5D. This suggests that also the particle Hamiltonian~\eqref{eq:GRpHam} has a 5D interpretation. Indeed this is the case; upon making the identifications
 \begin{align}
  P_0 = \frac{H}{c} \,, \quad P_5 = mc \,, \quad (\text{\it space-like}) \,, \label{spacelike}
 \end{align}
the Hamiltonian \eqref{eq:GRpHam} corresponds to geodesic motion along a null geodesic,
 \begin{align}
    g^{AB} P_A P_B = 0 \,.
 \end{align} 
The metric is the uplift of the Einstein-Maxwell-dilaton background
 \begin{align} 
    ds^2 = \eta_{AB} dx^A dx^B + \left( \frac{r_M}{r} \right)(dx^0 - dx^5)^2 \,, \label{bargmann-kepler}
 \end{align}
and is a so-called Bargmann space. Because of the harmonicity of the potential $\frac{r_M}{r}$, this space is Ricci-flat, i.e. it solves the vacuum Einstein equations~\cite{Duval:1990hj}. This spacetime is an example of a  pp-wave (plane-fronted wave with parallel rays) due to the presence of the covariantly constant null-vector $\partial_0 +\partial_5$.

Interestingly, the same 5D perspective can also be related to the non-relativistic Kepler case; this is referred to as the Eisenhart-Duval lift~\cite{Cariglia:2012fi,Cariglia:2014dwa}. Starting in 5D and making the alternative identifications
 \begin{align}
    \frac{P_0-P_5}{\sqrt{2}}= \frac{ \Bar{H}_{\rm Kep}}{c} \,, \quad \frac{P_0+P_5}{\sqrt{2}}= mc \,, \quad (\text{\it null}\,) \,, \label{null}
 \end{align}
the vanishing norm of the 5D momentum leads to a particular non-relativistic Kepler system
 \begin{align}
     \Bar{H}_{\rm Kep}= \frac{p^2}{2m}-\frac{r_M mc^2 }{r}\,.
 \end{align}
Though in different guises, the above two reduced systems therefore are one and the same, up to time reparametrisation (or conformal transformations, in the case of null-geodesics~\cite{Zhang:2019gtt}). Yet the reductions correspond to what appear to be distinct physical systems, dependent on the different assignments \eqref{spacelike} or \eqref{null} for mass and energy, as discussed in generality in Appendix~\ref{sec: prop1}.

The 5D formulation allows for a geometric perspective on the so(4) hidden symmetry of the Kepler system, as these can be seen to be generated by an interplay of Killing vectors and tensors~\cite{Duval:1990hj,Cariglia:2014dwa}. 
A vector field $K$ is a Killing vector field if 
\begin{equation}
    \nabla_{(a}K_{b)}=0\,.
\end{equation}
and generate conserved quantities $K^a p_a$ along geodesic motion. A rank-$n$ Killing tensor $K_{a_1\dots a_n}$ is defined analogously by~\cite{Lindstrom:2022nrm}
\begin{equation}
    \nabla_{(a_1}K_{a_2 \dots a_{n+1})}=0\,,
\end{equation}
and leads to a conserved quantity $K^{a_1\dots a_n}p_{a_1}\cdots p_{a_n}$ along geodesics~\cite{Walker:1970un}. Notice that the arrow of implication goes both ways: any conserved quantity that can be written as a tensor contracted with momenta also implies the existence of a Killing tensor. Conserved quantities of quadratic or higher order in momenta are symmetries of the phase space that cannot be constructed as the lift of a symmetry of configuration space, and are referred to as hidden symmetries~\cite{Cariglia:2014ysa}. 

The above concepts can be applied to the 5D incarnation of the Kepler problem. The Bargmann space~\eqref{bargmann-kepler} has three Killing vectors associated with angular momentum. In addition, it has three Killing tensors that generate the Laplace-Runge-Lenz vector reading
\begin{equation} \label{eq: 5dLRL}
    \vec{A}= (\vec{p}\cross (\vec{r}\cross \vec{p})) - \frac{1}{2} (P_0+P_5)^2 \frac{r_M c^2}{r}\vec{r} \,,
\end{equation}
which is quadratic in momenta.

Let us conclude this section with two remarks on the 5D symmetry algebra. Firstly, when we take the limit $M\to 0$ toward flat space, the second term in the above vector vanishes, implying the quadratic Killing tensors become reducible, as they are now built completely out of the Killing vectors related to the conserved angular momentum $\vec{L}=\vec{r}\cross \vec{p}$ and the momentum $\vec{p}$. The latter in this case is conserved as well, since the space regains translational symmetry in the three spatial directions of $\vec{r}$.

Secondly, note that none of the conserved quantities related through Noether's theorem to the so(4) symmetry of the Kepler problem is generated by a genuinely conformal Killing tensor or vector. All conserved functions commute with the 5D geodesic Hamiltonian regardless of its value. As such, it is clear that the relation between Kaluza-Klein dynamics and the Kepler system still holds for \emph{massive} geodesics, though the effect on the dynamics is of negligible consequence, as it results in the addition of a fixed constant to the Hamiltonians after reduction. % This constant is fixed in the same sense as the mass-like constant $P_v$ or $P_u$ in case of reduction over $v$ or $u$ respectively. That is to say, in the reduced systems it is part of the definition of the Hamiltonian, rather than a constant function on phase space. 

\section{Electromagnetism}

\noindent
We now turn to the second element of the class of Hamiltonians~\eqref{Ham} for which we propose an interpretation as coming from a relativistic field theory, being the case with a linear function $F(x)=f_1 x$. In this case the Hamiltonian reads
\begin{align}
     H^2 = m^2 c^4 + p^2 c^2 - \left( \frac{1}{r} \right)f_1 (m c^2 + H) \,. \label{linear}
 \end{align}
Upon making the space-like replacements~\eqref{spacelike}, this becomes
 \begin{align}
  \eta^{AB} P_A P_B -  \left( \frac{1 }{r} \right) \frac{f_1}{c}  (P_5 +P_0) = 0 \,.
 \end{align}
In contrast to the quadratic case, the above cannot be written in terms of a Lorentzian metric solving the Einstein equations\footnote{It can be written as a momentum norm on a Finsler space, a generalisation of pseudo-Riemannian space~\cite{Beil1987}.}. Instead, the natural interpretation of the lifted form is in terms of an electromagnetic force, linear in momenta. 

This can be understood in terms of the classical double copy, which maps the specific class of so-called Kerr-Schild backgrounds of General Relativity to solutions of Maxwell's equations in electromagnetism. Among more obvious solutions such as the Kerr and Schwarzschild black holes, also time-dependent pp-waves have been shown to have a classical double copy structure~\cite{Monteiro:2014cda}. The spacetime we are interested in is a subset of this, being a time-independent pp-wave.

The class of Kerr-Schild solutions takes the form
\begin{equation}
    g_{AB}=\Bar{g}_{AB}+r_s \phi\,  l_Al_B\,,
\end{equation}
and consist of a background metric $\Bar{g}$ (which should separately satisfy Einstein's equations) plus a deformation in terms of a harmonic scalar function, $\Box \phi(r) = 0$, and a vector $l_A$ that is null with respect to $\Bar{g}$ and therefore also total metric $g$. The Schwarzschild radius reads $r_s=2GM/c^2$. The Kerr-Schild double copy then says that the vector $A_A=\phi\,  l_A$ is a solution of the Maxwell equations and $\phi(r)$ is a solution of the scalar field equation.

The Bargmann space introduced above is exactly of Kerr-Schild form, with a flat background metric and SO(3) rotationally symmetric function, $\phi =2 / r$. This gravitational solution therefore generates, via the double copy procedure, the following solution of Maxwell's equations, now including all constants: 
\begin{equation} 
    A_A=\frac{g}{4\pi}\Tilde{q}\phi(r)l_A\,,\label{Maxwell} 
\end{equation}
where we introduced $g$ as electromagnetic coupling and $\Tilde{q}$ as a charge. This single copy is itself an electromagnetic pp-wave and $l_A$ is its wave vector. 

Given the interesting dynamics of geodesics in the Bargmann space-time, can we expect similar behaviour in the Maxwell configuration? In other words, does the hidden symmetry and maximal superintegrability survive the double copy? Interestingly, this turns out indeed to be the case: the Hamiltonian with linear force function \eqref{linear} has an interpretation as a charged particle in the single copy background \eqref{Maxwell}. 
Specifically, after setting $f_1=\frac{g^2}{4\pi}q\Tilde{q}$ and uplifting to 5D via \eqref{spacelike}, the Hamiltonian of the linear case can be written as 
 \begin{align}
  \eta^{AB} \tilde P_A \tilde P_B = 0 \,,\  \text{with }\tilde P_A=P_A - \frac{g q}{c} A_A\,,
 \end{align}
in terms of the modified momentum $\tilde P$ of charged particles in the vacuum EM background~\eqref{Maxwell}. This places the linear case on a par with the quadratic one, with interpretations in terms of electromagnetic and gravitational force fields, respectively.

With the 5D interpretation in hand, one can consider its dimensional reductions. The null reduction straightforwardly leads to the usual non-relativistic Hamiltonian, which in this case is interpreted as the Coulomb system. In contrast, the space-like reduction leads to a relativistic field theory with coupled degrees of freedom. In order to see this, separate the 5-dimensional vector field in parts $A_A=(A_\mu,A_5=\chi)$. The field Lagrangian then reads
\begin{equation}
    \mathcal{L}=-\frac{1}{4}F_{AB}F^{AB}=-\frac{1}{4}F_{\mu\nu}F^{\mu\nu}-\frac12(\partial_\mu \chi)(\partial^\mu \chi)\,,
\end{equation}
implying a theory with a scalar $\chi$ and vector $A_\mu$. The particle Lagrangian can be found by taking 
\begin{equation}
    \mathcal{L}_p=-\frac{1}{2 h}\eta_{AB}\dot{x}^A\dot{x}^B +\frac{g q}{c} A_A \dot{x}^A\,,
\end{equation}
which can be Legendre transformed to write it in terms of $P_5$ instead of $\dot{x}_5$. One then has
\begin{equation} %% Write more general, check signs.
    \mathcal{L}_p=\frac{1}{2h}\eta_{\mu\nu}\dot{x}^\mu\dot{x}^\nu-\frac{h}{2}\left(\frac{g q}{c}\chi-P_5\right)^2- \frac{g q}{c}A_{\mu}\dot{x}^\mu\,.
\end{equation}
Solving the equation for the auxiliary variable $h$, picking the branch giving the correct kinetic term, we find the Lagrangian
\begin{equation}
    \mathcal{L}_p|_h=-\sqrt{-\left(\frac{g q}{c} \chi- P_5\right)^2\eta_{\mu\nu}\dot{x}^\mu\dot{x}^\nu}- \frac{g q}{c} A_{\mu}\dot{x}^\mu\,.
\end{equation}
The solutions of the fields are, as in~\eqref{Maxwell}, given by
\begin{equation}
    \chi=\frac{g}{4\pi}\tilde q \frac{2}{r} \,,\qquad A_0=-\frac{g}{4\pi}\tilde q \frac{2}{r}\,.
\end{equation}
This is exactly the theory found by~\cite{Alvarez-Jimenez:2018lff}, describing a non-minimally coupled scalar-vector theory, having a LRL symmetry. It has its origin in $\mathcal{N}=4$ supersymmetric Yang-Mills, which similarly enjoys an additional symmetry in the the limit that one of the particles is much larger than the other~\cite{Caron_Huot_2014,Sakata:2017pue}.

Due to the 5D interpretation of this theory, we now have a natural link between the non-relativistic Coulomb system and this specific relativistic Maxwell-axion field theory. The relation between both symplectic reductions implies that the energy-level equivalence between the two theories necessarily extends to all orders in the PN expansion, with the time reparametrisation as discussed in Appendix~\ref{sec: prop1}.

\section{Nordstr\"{o}m gravity}

\noindent
The last Hamiltonian in the class~\eqref{Ham} that we propose a field theoretic interpretation for, is the one with function $F(x)=f_0$ constant: 
\begin{align} \label{eq:constantF}
     H^2 = m^2 c^4 + p^2 c^2 - \left( \frac{1}{r} \right) f_0  \,.
 \end{align}
The Hamiltonian of this system is a simple function of the Kepler Hamiltonian, and as such naturally has a hidden symmetry and closed orbits. Given the sequence of double and single copy in the previous two sections, however, it is natural to investigate its interpretation as a zeroth copy relativistic field theory. This will feature an attractive scalar field, as first proposed by Nordstr\"{o}m in 1912~\cite{Nordstrom_1912}.
 
The scalar `zeroth copy' is an Abelian version of the bi-adjoint scalar. The constant $f_0$ in the potential term now has the dimension of energy squared times length, so we set $f_0=r_M m_s^2c^4$, with $r_M$ as before and $m_s$ some mass. Lifting the Hamiltonian to 5D yields
\begin{align}
  \eta_{AB} P^A P^B = m_s^2 c^2 \frac{r_M}{r} \,,
\end{align}
and hence a position-dependent mass. Legendre transforming this to the Lagrangian of a geodesic yields
\begin{equation}
    \mathcal{L}_p=\frac{\eta_{AB} \dot{x}^A\dot{x}^B}{2h}+ m_s^2 c^2\frac{h}{2}\frac{r_M}{r}\,, \label{with-vielbein}
\end{equation}
where we have included the auxiliary Einbein $h$. It can be solved for to generate 
\begin{equation}
    \mathcal{L}_p|_h=m_sc\sqrt{\phi(r)\eta_{AB} \dot{x}^A\dot{x}^B}\,,
\end{equation}
with the scalar field $\phi(r)=\frac{r_M}{r}$ minimally coupled and now dimensionless, and solving the field equation in a vacuum. 

This completes the trilogy of double copy related relativistic field theories, at least from the 5D perspective: the Hamiltonians with constant, linear and quadratic functions correspond to the natural sequence of spin-$0,1,2$ exchange field theories.

Again, it will be interesting to investigate the dimensional reductions. As before, the null reduction directly generates a non-relativistic system akin to the Kepler and Coulomb ones. The space-like reduction, instead, generates a coupling to a new scalar field profile. After dimensional reduction of \eqref{with-vielbein} and subsequent elimination of the Einbein $h$, we find 
\begin{equation}
    \mathcal{L}_p|_h=-m_sc\sqrt{\Bar{\phi}(r) (-\eta_{\mu\nu}\dot{x}^\mu\dot{x}^\nu )}\,, \label{Nordstrom}
\end{equation}
with the scalar now given by 
 \begin{align}
  \Bar{\phi}(r)=1-\frac{r_M}{r}  \,. \label{profile}
 \end{align}
This particular profile for the scalar field then produces exactly closing time-like curves, and it is the Lagrangian giving rise to $H$ in equation~\eqref{eq:constantF} with proper time parametrisation.

In general, scalar fields generate a perihelion precession that is given to first post-Newtonian order by
\begin{equation}
    \Delta\theta=-\frac{1+a_2}{6}\Delta\theta_{\rm GR}\,,
\end{equation}
with $\Delta\theta_{\rm GR}$ the perihelion shift in GR and $a_2$ is a parameter that characterises the coupling of the scalar field (at quadratic level) to the probe particle~\cite{Deruelle:2011wu}.
How does this and  similar results relate to our system displaying no precession? Whereas $a_2=0$ in Nordstr\"{o}m's final theory (see \cite{Sundrum:2003yt}), $a_2=1$ in the only-scalar limit of~\cite{Davis_2023zqv} and Nordstr\"{o}m's first theory (see~\cite{Nordstrom_1912,Ravndal_2004ym}) and indeed, working out the square-root expansion we have $a_2=-1$ for the scalar~\eqref{profile}.   

At this point it is also interesting to compare to the classification of Bertrand spacetimes, as pioneered by Perlick~\cite{VPerlick_1992}. These spacetimes are defined by having the special property of having closed orbits; however, Bertrand spacetimes are not required to solve the source-free Einstein's equations and generically need to be supplemented with non-trivial energy-momentum tensors. In contrast, in all three cases that we have discussed, the backgrounds satisfy the source-free bulk equations of motion of the three relativistic field theories. The latter case of the zeroth copy includes the coupling to the harmonic scalar field in~\cite{Nordstrom_1912}. From the particle's perspective, this is equivalent to coupling to a conformally flat metric with overall factor~\eqref{profile}; this therefore has to be a Bertrand space-time. We have checked that the general class of solutions of~\cite{VPerlick_1992} indeed includes this as the unique conformally flat possibility\footnote{Specifically, the mapping is from the lower signs in equation (13) of~\cite{VPerlick_1992} with $D\to-2/r_M^2$, $K\to 4/r_M^4$ and $G\to 0$ and rescaling the time $t\to t\; r_M/\sqrt{2}$.}.

\section{Discussion and conclusion}

\noindent
In this paper, we have studied the possible interactions giving rise to Keplerian symmetry in field theories. We have presented the unique extension of the classical Kepler system to a relativistic system with the same so(4) algebra and the correct special relativistic limit, at least up to 5PN. This relativistic class of Hamiltonians contains terms naturally interpretable through a lift to 5 dimensions as spin-0,-1 and -2 interaction. We discussed how these project to relativistic theories in 4 dimensions, and their relation to the classical Kepler problem. Moreover, the interactions can be viewed as zeroth, single and double copy in a classical double copy structure. 

The 4D systems with so(4) symmetry corresponding to spin-1 and -2 in 5D are known in literature to stem from $\mathcal{N}=4$ super Yang-Mills and $\mathcal{N}=8$ supergravity \cite{Alvarez-Jimenez:2018lff, Caron_Huot_2019}, which can be truncated to the non-minimally coupled scalar-vector and Einstein-Maxwell-dilaton with $a=\sqrt{3}$ respectively. However, the spin-0 system in terms of a dilaton field theory, whose resulting effective metric coincides with a type of Bertrand spacetime~\cite{VPerlick_1992}, appears to have escaped attention so far.

The trajectories in all systems from a 5D perspective can be seen as reparametrisations of trajectories of a classical Kepler problem. In a similar vein, all one-body dynamics in the particular relativistic systems discussed can be linked to classical dynamics. Though we have focused on the particular case of relativistic 1-center problems with classical symmetry, all the above constructions hold with any harmonic background potential $\phi$: the Einstein, Maxwell and scalar field equations are still solved after the 5D lift. The reduced systems are then through time reparametrisation related to their classical counterparts. 

For example, a multi-center background in Einstein-Maxwell-dilaton theory with $a=\sqrt{3}$ and extremal objects can be constructed, which is orbited by an anti-extremal test particle. This is the space-like reduction of a Bargmann spacetime of which the null-reduction is the classical Newtonian multi-center system. The relevant harmonic function for $n$ centers (all extremally charged) is given by $-\phi(\vec{r})=\frac{r_{M_1}}{\abs{\vec{r}-\vec{r}_1}}+\dots+\frac{r_{M_n}}{\abs{\vec{r}-\vec{r}_n}}$~\cite{Gibbons:1993dq}, and the Hamiltonian for an anti-extremal test particle of mass $m$ satisfies
\begin{equation}
    H^2=m^2c^4+p^2c^2+\phi(\vec{r})\left(mc^2+H\right)^2\,.
\end{equation}
This connection immediately implies these solutions possess the same integrability properties as their equivalents in non-relativistic mechanics, that is, the two-center case is integrable, while the systems with a higher number of centers all display chaotic scattering~\cite{MR0749024,knauf2004integrability,Knauf2002TheNP}, as discussed for more values of the dilaton coupling in~\cite{Martijn}.

Some further questions present themselves. Firstly, all mentioned so(4) preserving systems are test-particle limits. It would be interesting to consider the two-body problem beyond the first order in mass ratio or relativistic expansions in these theories, as done for example by~\cite{Parra_Martinez_2020dzs}, who still found no precession of the orbits at third Post-Minkowskian order for extremal black holes in $\mathcal{N}=8$ supergravity (but did not confirm or exclude persistence of so(4) symmetry). Modern scattering amplitude methods\footnote{For instance, it would be interesting to investigate whether the effective one-body approach in Kerr-Schild formulation \cite{Ceresole} allows to capture comparable mass ratio effects in the 5D Kerr-Schild potential.} seem to be well-suited to this task, especially in the light of the double copy, which might allow one to directly export results from $\mathcal{N}=4$ super Yang-Mills to $\mathcal{N}=8$ supergravity.

Secondly, is it possible to extend the idea of relating relativistic systems to geodesics in higher dimensional spaces to multi-body systems? This would need to account for the multi-worldline nature of such systems, possibly by including more time-like dimensions in the higher dimensional system. It would be intriguing to see if this alternative perspective could lead to new or simplifying insights into relativistic dynamics, and whether interesting connections can be made to classical systems. 

\section*{Acknowledgements}

\noindent
We are grateful to Cliff Burgess, Joaquim Gomis, Johannes Lahnsteiner, Scott Melville, Fernando Quevedo and Ziqi Yan for stimulating discussions. D.N.~is supported by the Fundamentals of the Universe research program within the University of Groningen.

\appendix

\section{Relativistic dynamics as reparametrised classical trajectories} \label{sec: prop1}

\noindent
In this more mathematically oriented appendix we provide further details on the time reparametrisation that relates the (non-)relativistic Hamiltonians, and make this connection rigorous in terms of the Marsden-Weinstein reduction of the higher-dimensional phase space. 

To this end, we will start from a slight generalisation of the above, with Hamiltonians $H_2(q,p)$ (introduced in \cite{deNeeling:2023egt}) solving an equation of the form
\begin{equation} \label{eq:implicitform}
    -f(H_2(q,p))= \frac{p^2}{2C_u}+\frac{g(H_2(q,p))}{C_u}\Phi(q)\,,
\end{equation}
with $f$ and $g$ suitable smooth functions $\R\to\R$ and $C_u$ a nonzero constant. We will show that these are similarly related to classical Hamiltonians using Marsden-Weinstein reduction, employing a symmetry of a Hamiltonian system to reduce its dimension by an even number. This procedure guarantees that the reduced space remains symplectic, since the symplectic form is given by the restriction of the initial symplectic form to the new space~\cite{MARSDEN1974121}.

\begin{assumptions} \label{assump1}  
    \begin{enumerate}
        \item Let $\Phi:\R^d\backslash \Delta \to \R$ be a smooth function, $\Delta$ a discrete set of isolated points and $\pder[\Phi(q)]{q}\neq 0$ for all $q\in \R^d\backslash \Delta$.
        \item Let $H_1:M\to \R$ be a Hamiltonian on $M=T^*(\R^d\backslash \Delta)$ with coordinates $q\in\R^d$ and $p\in\R^d$, defined by
        \begin{equation} \label{eq:Ham1}
            H_1(q,p)= \frac{p^2}{2 f(C_v)}+\frac{g(C_v)}{f(C_v)}\Phi(q)\,,
        \end{equation}
        for two smooth functions $f,g:\R\to \R$.
        \item Let $C_u, C_v\in \R$ be nonzero constants such that $f(C_v)\neq 0$, $g(C_v)\neq 0$ and  
        \begin{equation}
            g'(C_v)\Phi(q)+f'(C_v) C_u\neq 0\,,\ \text{for all } q\in \R^d\backslash \Delta.
        \end{equation}
    \end{enumerate}
\end{assumptions}

\begin{proposition} \label{newprop}
Given Assumptions~\ref{assump1}(1--3), if the Hamiltonian $H_2(q,p)$ implicitly defined by~\eqref{eq:implicitform} exists for $H_2(q,p)=C_v$, the trajectories on an energy level $H_1=-C_u$ of the first Hamiltonian and $H_2=C_v$ of the second are in one-to-one correspondence, up to the time reparametrisation 
\begin{equation}
    \der[t_1]{t_2}=\frac{f(C_v)}{g'(C_v)\Phi(q)+f'(C_v) C_u}\,.
\end{equation}
    
\end{proposition} 
\begin{proof}
Consider a manifold $\Bar{M}=T^*(\R^2\cross(\R^d\backslash \Delta))$, with canonical coordinates $(u,v,q;P_u,P_v,p)$ and Hamiltonian $\mathcal{H}:\Bar{M} \to \R$ given by
\begin{equation}
    \mathcal{H}(u,v,q;P_u,P_v,p)=g(P_v)\Phi(q)+f(P_v)P_u+\frac{p^2}{2}\,.
\end{equation}
The proposition is an application of Marsden-Weinstein reduction restricted to the level set $\mathcal{H}^{-1}(0)$. Since $u$ and $v$ are cyclic coordinates by construction and the corresponding actions (by translation) are both free and proper, we can apply the Marsden-Weinstein theorem~\cite{MARSDEN1974121} to symplectically reduce over either of them.
We can construct the two possible reductions using the momentum maps $\mu_{\star}:=(\mathcal{H},P_{\star}):\Bar{M} \to \mathbb{R}^2$, where $\star$ denotes either $u$ or $v$, giving two, in principle different base spaces $B_{\star}=\mu_{\star}^{-1}(0,C_{\star})/\mathbb{R}_{\star}$ and corresponding projections $\pi_{\star}:\mu_{\star}^{-1}(0,C_{\star})\to B_{\star}$.
Here $(0,C_{\star})$ are assumed to be regular values of $\mu_{\star}$, for now. Note that the base spaces are odd-dimensional, as the value of $\mathcal{H}$ is fixed, so that the established relation is between energy levels of $H_1$ and $H_2$. 

Consider the regularity requirements $\dif \mathcal{H}\neq 0$ and $\dif P_{u},\dif P_v\neq 0$. The latter are trivially satisfied, while the former gives 
\begin{align}
    &\dif \mathcal{H}(u,v,q;P_u,P_v,p) =\pder[\Phi]{q} g(P_v) \dif q +  f(P_v)\dif P_u \notag \\
    &\qquad + ( g'(P_v) \Phi(q) + f'(P_v)P_u)\dif P_v + p\, \dif p \,.  
\end{align}
By the Assumptions~\ref{assump1}(1--3) this is never vanishing on the domain of consideration, nor are $\dif \mathcal{H}$ and $\dif P_{u},\dif P_v$ linearly dependent. This shows that the values $(0,C_{\star})$ are indeed regular.

Let us now reduce the system over the cyclic coordinates $(u.v)$ to obtain new Hamiltonian functions that are exclusively dependent on the classical variables $(q,p)$. When we reduce by $v$ we have  
\begin{equation}
    g(C_v)\Phi+f(C_v)P_u+\frac{p^2}{2}=0\,,
\end{equation}
so that reparametrisation $\lambda\to f(C_v)\lambda$ results in
\begin{equation}
    -P_u=\frac{p^2}{2f(C_v)} +\frac{g(C_v)}{f(C_v)}\Phi(q)\,.
\end{equation}
This means that we effectively have a Hamiltonian $-P_u=H_1(q,p)$, which moreover has a time parameter given by $t_1=u=f(C_v)\lambda$ associated to it, as is apparent from Hamilton's equation $\der[u]{\lambda}=\pder[\mathcal{H}]{P_u}$. 

If, instead, we reduce phase space $\Bar{M}$ over $u$, we end up with 
\begin{equation}
    -f(P_v)= \frac{p^2}{2C_u}+\frac{g(P_v)}{C_u}\Phi(q)\,.
\end{equation}
By assumption, this is solved by $H_2(q,p)=C_v$. Hamilton's equation  
\begin{equation}
    \der[v]{\lambda}=g'(P_v)\Phi+f'(P_v)C_u\,,
\end{equation}
then implies the associated time $t_2$ is given by $\dif t_2=\dif v=(g'(P_v)\Phi+f'(P_v)C_u)\dif \lambda$. For energy levels $H_1=-C_u$ of the first Hamiltonian and $H_2=C_v$ of the second, this means the reparametrisation from the first Hamiltonian to the second is given by
\begin{equation}
    \der[t_1]{t_2}=\frac{f(C_v)}{g'(C_v)\Phi(q)+f'(C_v) C_u}\,,
\end{equation}
which concludes the proof of the relation between the two.
\end{proof}

\noindent
For concreteness we now apply this procedure to the relativistic class of Hamiltonians
 \begin{align} \label{eq:relHams}
     H^2 = m^2 + p^2 - \left( \frac{1}{r} \right)  F\left(m + H\right) \, 
 \end{align}
dropping factors of $c^2$ for legibility. As discussed in the main text, under the replacements~\eqref{spacelike} these lift up to the level set $\mathcal{H}^{-1}(0)$ of 
\begin{equation}
    \mathcal{H}=P_5^2-P_0^2+p^2-\left( \frac{1}{r} \right) F\left(P_5+P_0\right)\,. 
\end{equation}
Hamilton's equation tells us that the relation between the time-parameter $\lambda$ of this new Hamiltonian is related to the old time parameter $t$ by 
\begin{equation}
    \der[t]{\lambda}=-2P_0-\left( \frac{1}{r} \right) F'\left(P_5+P_0\right)\,,
\end{equation}
where $t$ is conjugate to $P_0$ (and hence $H$).

The canonical transformation to null coordinates 
\begin{align}
    (x_+,x_-)&=\tfrac{1}{\sqrt{2}}\left(x_5+t,x_5-t\right)\,,\\
    (P_+,P_-)&=\tfrac{1}{\sqrt{2}} (P_5+P_0,P_5-P_0)\,, 
\end{align}
results in
\begin{equation}
    -2P_-P_+=p^2-\left( \frac{1}{r} \right) F\left(\sqrt{2}P_+\right)\,
\end{equation}
on level set $\mathcal{H}^{-1}(0)$. Assuming $P_+ \neq 0$, a time reparametrisation to $\tilde \lambda= 2 P_+ \lambda$ reduces it further to
 \begin{equation}
    -P_-=\frac{p^2}{2 P_+}-\left( \frac{1}{r} \right) \frac{1}{2 P_+}  F\left(\sqrt{2}P_+\right)\,,
\end{equation}
so the alternative symplectic reduction over $x_+$ setting $P_+=m_+$ results in $-P_-=:H_-$ defining a Kepler Hamiltonian reading
\begin{equation}
    H_-=\frac{p^2}{2 m_+}-\left( \frac{1}{r} \right) \frac{1}{2 m_+}  F\left(\sqrt{2}m_+\right)\,.
\end{equation}
The relation between both time parametrisation therefore reads
\begin{equation}
    \der[x_-]{\lambda}=\pder[\mathcal{H}]{P_-}=2P_+\,,
\end{equation}
where $x_-$ is conjugate to $H_-$. 

Crucially, the Marsden-Weinstein result guarantees that both reductions preserve the symplectic structure.
From the above, it directly follows that the time reparametrisation that links trajectories on energy levels $H_-=\frac{1}{\sqrt{2}}(H_0-m)$ of the Kepler Hamiltonian to those on energy levels $H=H_0$ of the original is given by
\begin{equation}
    \der[x_-]{t}=\frac{\sqrt{2}(H_0+m)}{-2H_0-\left( \frac{1}{r} \right) F'\left(m+H_0\right)}\,.
\end{equation}
This transformation maps relativistic systems that violate Kepler's third law to Kepler ones that satisfy it. 

Note that only when $F(x)$ is homogeneously quadratic, does each energy level map to the same Kepler background, such that the potential coefficient in the Hamiltonian becomes linear in the mass of the probe. For all other functions, the effective Newton's constant becomes probe-mass dependent and hence would violate the equivalence between gravitational and inertial mass.

%\printbibliography
\bibliography{name}

\end{document}